\documentclass[12pt,a4paper]{amsart}
\usepackage[margin=2.5cm]{geometry}

\usepackage[utf8]{inputenc}
\usepackage{amsmath,amssymb,amsfonts,amsthm,amsopn,amscd}
\usepackage{dsfont}
\usepackage{graphicx}
\usepackage{longtable}
\usepackage{titletoc}
\usepackage[section]{placeins}
\usepackage{thm-restate}
\usepackage{booktabs}\renewcommand{\arraystretch}{1.2}
\usepackage{kbordermatrix}
\usepackage{enumitem}

\usepackage{mathtools}
 \usepackage{bm}
 
\usepackage{upgreek}
\usepackage{xcolor}
\usepackage{nameref}
\usepackage[breaklinks]{hyperref}
\hypersetup{
    colorlinks,
    linkcolor={red!40!black},
    citecolor={blue!50!black},
    urlcolor={blue!20!black}
}

\DeclareMathOperator{\tr}{tr}

\DeclareMathOperator{\wt}{wt}

\newcommand{\bra}[1]{\mathinner{\langle #1|}}
\newcommand{\ket}[1]{\mathinner{|#1\rangle}}

\newcommand{\ot}[0]{\otimes}
\newcommand{\one}[0]{\mathds{1}}
\renewcommand{\a}{\alpha}

\usepackage{thmtools, thm-restate}
\newtheorem{theorem}{Theorem}
\newtheorem*{theorem*}{Theorem}
\newtheorem{proposition}[theorem]{Proposition}

\newtheorem{example}[theorem]{Example}
\newtheorem{remark}[theorem]{Remark}
\newtheorem{problem}{Open Problem}
\newtheorem*{problem*}{Open Problem}

\newtheorem*{result*}{Result}

\newcommand{\N}{\mathds{N}}
\newcommand{\R}{\mathds{R}}
\newcommand{\C}{\mathds{C}}
\newcommand{\Q}{\mathds{Q}}
\newcommand{\E}{\mathcal{E}}

\newcommand{\II}{\mathcal{I}}

\newcommand{\nn}{\nonumber}

\begin{document}

\title[
SDP bounds on quantum codes: rational certificates
]{
SDP bounds on quantum codes: rational certificates}

\date{\today}

\author{Gerard Anglès Munné}
\address{
Gerard Anglès Munné,
Faculty of Mathematics, Physics and Informatics,
University of Gdańsk,
Wita Stwosza 57, 80-308 Gdańsk, Poland
}

\author{Felix Huber}
\address{
Felix Huber,
Division of Quantum Information,
Faculty of Mathematics, Physics and Informatics,
University of Gdańsk,
Wita Stwosza 57, 80-308 Gdańsk, Poland
}
\email{felix.huber@ug.edu.pl}

\thanks{
We thank 
Markus Grassl,
David de Laat, 
Nando Leijenhorst,
and
Andrew Nemec
for discussions.
GAM was supported by NCN grant no. 2024/53/B/ST2/04103.
FH was funded in whole or in part by the National Science Centre, Poland 2024/54/E/ST2/00451
and by the Polish National Agency for Academic Exchange under the Strategic Partnership Programme grant BNI/PST/2023/1/00013/U/00001.
For the purpose of Open Access,
the author has applied a CC-BY public copyright licence to any Author Accepted Manuscript (AAM) version arising from this submission. We acknowledge the use of a computational server financed by the Foundation for Polish Science (IRAP project, ICTQT, contract no. 2018/MAB/5/AS-1, co-financed by EU within Smart Growth Operational Programme).
}

\begin{abstract}
A fundamental problem in quantum coding theory is to 
determine the maximum size of quantum codes 
of given block length and distance. 
A recent work 
introduced bounds based on semidefinite programming, strengthening the well-known quantum linear programming bounds.
However, floating-point inaccuracies prevent the extraction 
of rigorous non-existence proofs
from the numerical methods.
Here, we address this by providing rational infeasibility certificates for a range of quantum codes.
Using a clustered low-rank solver with heuristic rounding to algebraic expressions, we can improve upon $18$ upper bounds on the maximum size of 
$n$-qubit codes with $6 \leq n \leq 19$.
Our work highlights the practicality and scalability 
of semidefinite programming for quantum coding bounds. 
\end{abstract}

\maketitle

\section{Introduction}
Quantum error correcting codes are widely thought to play an integral part in future quantum computers, 
making the computation reliable by protecting quantum information from environmental noise and decoherence.
A long-standing problem is to understand for what parameters a quantum code exists. 
In particular, we ask the question: 
given block length $n$ and distance $\delta$,
what is the maximum code size $K$ for which a 
$(\!(n,K,\delta)\!)_2$ quantum code exists?
Similar to the classical situation, 
both analytical~\cite{ekert1996errorcorrectionquantumcommunication,681315,PhysRevA.69.052330,PhysRevA.81.032318,PhysRevLett.118.200502,Huber2020quantumcodesof,wei2026theorylowweightquantumcodes} and linear programming (LP) methods~\cite{PhysRevLett.78.1600,681315,796376,wei2026theorylowweightquantumcodes} provide upper bounds on the parameters of quantum codes.
At the same time, explicit 
and randomized constructions 
serve as lower bounds~\cite{
quant-ph/9705052v1,  
PhysRevLett.79.953,
PhysRevLett.78.405,
rains1997quantumcodesminimumdistance,
PhysRevLett.99.130505,
Hastings2014,
Rigby_2019,
Breuckmann_2021,
Panteleev_2022}.
Tables of upper and lower bounds on achievable parameters can be found in Refs.~\cite{table,Rigby_2019, HuberWyderka:ametable, Huber2020quantumcodesof}.

Inspired by semidefinite programming  (SDP) bounds for classical codes~\cite{Vallentin2019}, 
a recent work by the authors introduced a similar machinery for quantum codes~\cite{munne2025sdpboundsquantumcodes}.
However, both linear and semidefinite programming are limited by the fact that they operate with floating-point arithmetic.
The resulting approximate infeasibility certificates then yield coding bounds that are, at best,
formally unproven but correct;
while, at worst, give entirely spurious bounds due to numerical errors.

Here, we rectify this situation by introducing rigorous certificates that act as formal proofs of coding bounds.
The key idea is to use numerical semidefinite programs from which one can extract rational certificates:
a method introduced by Ref.~\cite{PaPe08} and 
refined since in a sequence of papers~\cite{
KALTOFEN2008359,
doi:10.1137/20M1351692,
cohn2024optimalitysphericalcodesexact,
MAGRON2021221} and which is related to central questions in semialgebraic geometry~\cite{Scheiderer2016, LAPLAGNE2024107631}.
These methods have now also found their way to quantum information theory~\cite{naceur2025certifiedboundsoptimizationproblems}.

In this work, we employ a low-rank solver with a heuristic rounding method by Leijenhorst et al.~\cite{Leijenhorst_2024,cohn2024optimalitysphericalcodesexact}, 
available as a Julia package at~\url{https://github.com/nanleij/ClusteredLowRankSolver.jl}. This solver finds an approximate numerical solution, which is then rounded to an exact solution in the field of algebraic numbers. 
This results in rigorous proofs of the three coding bounds 
found by Ref.~\cite{munne2025sdpboundsquantumcodes}, 
as well as improves upon $18$ further cases 
with block length $6 \leq n \leq 19$. 
All infeasibility certificates can be found at \url{https://github.com/ganglesmunne/SDP_bounds_on_quantum_codes_rational_certificates}.

\section{Linear programming bounds on quantum codes}

A qubit quantum code encodes a $K$-dimensional Hilbert space into a subspace of $n$-qubit Hilbert space. 
A code is said to have a distance $\delta$ if every
error that acts on at most $\delta-1$ subsystems can be detected or 
acts as identity on the code space.
Consequently, the code can corrects errors that affect at most 
$\lfloor (\delta-1)/2 \rfloor$ subsystems.
A qubit quantum code can then be characterized by its parameters $(\!(n,K,\delta)\!)_2$, 
where $n$ is the block-length, $K$ the size of the code, and $2$ stands for qubits.
Let $\Pi$ be the projector onto a $K$-dimensional code subspace. Denote by $\E_n$ a matrix basis composed of $n$-fold tensor products of the three Pauli matrices 
\begin{equation}
    X = \begin{pmatrix}
        0 & 1 \\
        1 & 0
    \end{pmatrix}\,,\quad
   Y = \begin{pmatrix}
        0 & -i \\
        i & 0
    \end{pmatrix}\,,\quad
   Z = \begin{pmatrix}
        1 & 0 \\
        0 & -1
    \end{pmatrix}\,, 
\end{equation}
and the $2\times 2$ identity matrix $\one$. 
We define the weight $\wt(E)$ of an element $E \in \E_n$ as the number of coordinates 
$E$ acts non-trivially on.
The Knill-Laflamme conditions then give a complete characterization of quantum error-correcting codes~\cite{PhysRevA.55.900}: 
Let $\Pi$ be a projector onto a code subspace of rank $K$ acting on $(\C^2)^{\ot n}$.
The code $\Pi$ has distance $\delta$, if and only if
for all $E \in \E_n$ with $0\leq\wt(E)<\delta$, 
\begin{align}
\Pi E \Pi = c_E\Pi\,,
\end{align}
holds where $c_E \in \C$ depends only on $E$. 
A code with $c_E = \tr(E)/2^n$ for all $E$ with 
$0 < \wt(E) < \delta$ is called {\em pure}.

The Knill-Laflamme conditions can also be formulated in the following way~\cite{681316}: A projector $\Pi$ of rank $K$ corresponds to a code with distance $\delta$, if and only if
\begin{align}\label{eq:klaflam}
    KB_j(\Pi)=A_j(\Pi) \quad\quad \text{for} \quad 0 \leq j\leq \delta-1 \,,
\end{align}
where 
\begin{align}\label{eq:enum}
A_j(\Pi) = \sum_{\substack{E\in \E_n \\ \wt(E) = j}} \tr(E\Pi) \tr(E^\dagger\Pi)\,, \quad\quad  B_j(\Pi) =\sum_{\substack{E\in \E_n \\ \wt(E)=j}} \tr(E\Pi E^\dagger\Pi)  \,.
\end{align}
are the quantum weight enumerators.

The two enumerators are linearly related by the {\it quantum MacWilliams identity}~\cite{PhysRevLett.78.1600,681316},
\begin{align}\label{eq:macwill}
	B_j(\Pi) & = 2^{-n}\sum^n_{i=0} K_j(i;n)A_i(\Pi)\,,
\end{align}
where $K_j(i;n)$ are the quaternary Krawtchouk polynomials~\cite{macwilliams1977theory}, 
\begin{align}\label{eq:krawpoly}
		K_j(i;n) = \sum^n_{\alpha=0} (-1)^\a 3^{j-\a}
		\binom{i}{\alpha}  \binom{n-i}{j-\alpha}     \,.
\end{align}
In terms of the enumerator polynomials $A(x,y) = \sum_{j=0}^n x^{n-j} y^j A_j$ and likewise for $B(x,y)$, the identity reads
\begin{equation}
    B(x,y) = A\Big(\frac{x+3y}{2}, \frac{x-y}{2}\Big)\,.
\end{equation}
The enumerators satisfy
\begin{align}
    A_0(\Pi) &= K^2\,,\nn\\
    A_j(\Pi) &\geq 0 \quad \quad \quad\quad \,\text{for} \quad 0 \leq j \leq n \,,\nn \\ 
    KB_j(\Pi) &\geq A_j(\Pi) \quad\quad \text{for} \quad 0\leq j\leq  n \,,
\end{align}
with equality in the last expression for $0\leq j \leq \delta-1$.

The Shadow enumerator
is also non-negative~\cite{796376},
\begin{align}\label{eq:shadow}
 S_j(\Pi)& = 2^{-n}\sum^n_{i=0} (-1)^i K_j(i;n)A_i(\Pi) \geq 0 \,,
\end{align}
for $0 \leq j \leq n$. It can also be expressed in terms of a polynomial transform,
\begin{equation}
    S(x,y) = A\Big(\frac{x + 3y}{2}, \frac{y-x}{2}\Big)\,.
\end{equation}

To determine whether a quantum code with parameters $(\!(n,K,\delta)\!)_2$ does not exist, one can formulate the following linear program (LP).
\begin{align}\label{eq:lp}
		\text{find} 	\quad & \{A_j\}^{n}_{j=0} \nn\\
		\text{subject to}	\quad & A_0=K^2\,, \nn\\
		& A_j, S_j \geq 0 \quad\quad \text{for} \quad 0\leq j\leq n \,, \nn \\
        &  KB_j \geq A_j  \quad\quad \text{for}\quad \delta \leq j\leq n \,, \nn\\
        & KB_j = A_j  \quad\quad \text{for}\quad 0 \leq j\leq \delta-1 \,,
\end{align}
and where $B_j$ and $S_j$ are linearly related to $A_j$ via Eqs.~\eqref{eq:macwill} and \eqref{eq:shadow}, respectively.
The linear program in Eq.~\eqref{eq:lp} can be used to obtain upper bounds on the allowed parameters of quantum codes.
If the LP in Eq.~\eqref{eq:lp} is infeasible, 
no code with the given parameters exists.

Additional constraints restrict the LP in Eq.~\eqref{eq:lp} in 
certain cases further.
\begin{enumerate}[label=(\alph*)]
\item {\bf Pure codes} are those for which detectable errors result in states that are orthogonal to the unaffected state, 
$\bra{i_L} E \ket{i_L} = 0$ for all $E$ with $0 < \wt(E) <\delta$ and all logical code basis $\ket{i_L}$. 
They satisfy
\begin{align}\label{eq:pureLP}
    A_j &= 0 \quad\quad \quad  \text{for} \quad 0 < j < \delta \,,
\end{align}
 
\item {\bf Self-dual  codes~\footnote{Not to be confused with self-dual CSS codes.}} are defined to be pure codes with $K=1$~\footnote{Quantum codes with $K=1$ fulfill the Knill-Laflamme conditions trivially. 
For this reason, they are additionally assumed to be pure, 
which implies that $\tr(\Pi E)=0$ for $0< \wt(E)<\delta $ with $E\in \E_n$.}.
Thus they satisfy Eq.~\eqref{eq:pureLP} and additionally~\cite{796376},
\begin{align}\label{eq:addLP}
 S_j &=0 \quad\quad \quad  \text{if} \quad  n-j \quad \text{is odd}\,, \nn\\
B_j &= A_j \hspace{-0.06cm}\quad\quad\,\,\, \text{for} \quad 0 \leq j \leq n\,.
\end{align}

\item {\bf Additive (or stabilizer) codes} with parameters $[\![n,k,\delta]\!]_2$ encodes $k$ to $n$ qubits, and thus $K=2^k$ with $k\in \N$.
For additive codes 
$A_j\in \mathbb{N}_+$
and
one of the following two cases holds~\cite{681315},
\begin{align}\label{eq:add_type}
	\sum_{j\geq0} A_{2j} =
	\begin{cases}
2^{n-k-1} \quad \text{(Type I)} \\  2^{n-k} \quad \quad \!\text{(Type II)}\,.
	\end{cases}
\end{align}
We use parameters $[\![n,k,\delta]\!]_2$ for additive codes and $(\!(n,K,\delta)\!)_2$ 
for general (e.g. non-additive) quantum codes. 
\end{enumerate}

\begin{example}\label{ex:enum}
The stabilizer group of the five-qubit additive code $[\![5,1,3]\!]_2$ 
is generated by 
\begin{equation}
 S= \langle XZZXI, IXZZX, XIXZZ, ZXIXZ \rangle  \,.
\end{equation}
Its code projector is
\begin{equation}
\Pi = \frac{1}{2^{n-k}} \sum_{s \in S}  s\,.
\end{equation}
from which we can compute the weight distribution $A(\Pi) = (4, 0, 0, 0, 60, 0)$.

\end{example}

\section{Semidefinite Programming bounds on quantum codes}

We sketch the semidefinite programming (SDP) bounds as derived by Ref.
~\cite[Theorem 22]{munne2025sdpboundsquantumcodes}. 
The formulation is inspired by the fact that the Lovász theta number upper bounds classical codes~\cite{1055985}, 
and corresponds to an intermediate level of a complete hierarchy presented in Ref.~\cite[Theorem 7]{munne2025sdpboundsquantumcodes}.

Below, we present three SDP formulations:
the main SDP, a Lov\'asz type SDP, and the symmetry-reduced version of
the main SDP. For the numerical results in Section~\ref{sec:cert}, 
we only use the symmetry-reduced SDP.

\subsection{Main SDP}The key idea is the following: Let $\varrho=\Pi/K$ be a state proportional to the projector of a code subspace $\Pi$. 
Let $\Gamma$ be a matrix which is indexed by all elements of the $n$-qubit Pauli basis 
$\E_n$, with entries determined by the pair $E_a, E_b \in \E_n$ as
\begin{align}\label{eq:moment}
    \Gamma_{ab}(\varrho)=\tr(E^\dagger_a \varrho) \tr(E_b \varrho) \tr(E_aE^\dagger_b\varrho) \,.
\end{align}
By convention, we set $E_0=\one$ and thus $\Gamma_{00}(\varrho)=1$.
One can show that $\Gamma(\varrho)$ is positive semidefinite, 
because it can be constructed as a moment matrix from $\varrho\succeq0$. 
Importantly, the quantum weight enumerators $A_j$~ [Eq.~\eqref{eq:enum}] can be obtained from $\Gamma$ by summing over its diagonal elements,
\begin{align}\label{eq:AfromGamma}
 A_j(\varrho) =\sum_{\substack{E_a\in \E_n \\ \wt(E_a) = j}} \Gamma_{aa}(\varrho) \,.
\end{align}
From $A_j(\varrho)$ one can obtain $B_j(\varrho)$ through the
quantum MacWilliams identity [Eq.~\eqref{eq:macwill}].
One can thus write the Knill-Laflamme conditions in Eq.~\eqref{eq:klaflam} in terms of the entries of the moment matrix $\Gamma(\varrho)$ as
\begin{align}\label{eq:klred}
	\frac{K}{2^n}\sum^n_{i=0} K_j(i;n) \!\!\sum_{\substack{E_{a} \in \E_n \\ \wt(E_a)=i}} \Gamma_{aa}(\varrho) \quad
	= \sum_{\substack{E_b \in \E_n \\ \wt(E_b)=j}} \Gamma_{bb}(\varrho) \quad\quad \text{for} \quad 0 \leq j\leq \delta-1  \,.
\end{align}
 
The matrix $\Gamma(\varrho)$ satisfies several additional conditions,
which follow from the fact that $\Gamma$ 
is constructed from a quantum state that is proportional a the projector of rank $K$~\cite[Section 5]{munne2025sdpboundsquantumcodes}. 
Considering only the real part of $\Gamma(\varrho)$, 
a necessary condition for an
$(\!(n,K,\delta)\!)_2$ quantum code to exist
is the feasibility of the following semidefinite program,
\begin{align}\label{eq:sdpgamma}
	\text{find} 		\quad & \Gamma \succeq 0\,, \nn \\
	\text{subject to} 	\quad
	& \Gamma_{00}=1 \,, \nn\\
        &\Gamma_{ab}=0 \hspace{-0.04cm}\quad\quad\,\,\text{if} \quad  E^\dagger_a E_b = - E_b E^\dagger_a \,, \nn \\
	&	\Gamma_{ab}=\Gamma_{cd} \,\,\quad \text{if} \quad  E_a E^\dagger_b  = E^\dagger_b E_a
	\quad \text{and}
	\quad (E^\dagger_a,E_b,\omega_{ab}E_a E^\dagger_b)\quad\, \nn\\ 
    \quad & \hspace{3.3cm} \text{is a permutation of} \,\,(E^\dagger_{c},E_{d},\omega_{cd}E_{c}E^\dagger_{d})\,, \nn\\
	& \sum^N_{b=0}\Gamma_{ab}=\frac{2^{n}}{K} \Gamma_{aa}\,, \nn \\
	& \frac{K}{2^n}\sum^n_{i=0} K_j(i,n) \!\!\sum_{\substack{E_{a} \in \E_n \\ \wt(E_a)=i}} \Gamma_{aa}
	\,\,\, = \,\,\,
	\sum_{\substack{E_b \in \E_n \\ \wt(E_b)=j}} \Gamma_{bb}  \quad \,\text{if} \quad 0\leq j \leq \delta-1 \,,
\end{align} 
where $N=4^n-1$ and $\omega_{ab} \in \{\pm 1,\pm i\}$.
If the SDP in Eq.~\eqref{eq:sdpgamma} is infeasible, 
then a code with parameters $(\!(n,K,\delta)\!)_2$ does not exist. 

Additional constraints can be imposed to restrict the SDP~\eqref{eq:sdpgamma}  further.
\begin{enumerate}  [label=(\alph*)]
\item {\bf Pure codes.} 
One can additionally impose
\begin{align}\label{eq:selfdual}
\Gamma_{ab}=0
	\quad\quad 
	\text{for} \quad 0 < \wt(E^\dagger_a E_b)<\delta \quad \text{or} \quad 0 < \wt(E^\dagger_a)<\delta \quad \text{or} \quad  0 < \wt(E_b)<\delta\,.
\end{align}
\item {\bf Self-dual codes.} The conditions in (a) apply, together with

\begin{align}\label{eq:addLP_matrix}
 \sum^n_{i=0} (-1)^i K_j(i,n) \!\!\sum_{\substack{E_{a} \in \E_n \\ \wt(E_a)=i}} \Gamma_{aa} &=0 \quad\quad \quad  \text{if} \quad  n-j \quad \text{is odd}\,, \nn\\
\frac{K}{2^n}\sum^n_{i=0} K_j(i,n) \!\!\sum_{\substack{E_{a} \in \E_n \\ \wt(E_a)=i}} \Gamma_{aa}
	\,\,\, &= \,\,\,
	\sum_{\substack{E_b \in \E_n \\ \wt(E_b)=j}} \Gamma_{bb}  \quad \,\text{for} \quad 0\leq j \leq n \,,
\end{align}

\item {\bf Additive codes.} one can further impose that
\begin{align}
\Gamma_{ab}\geq 0 \quad \text{and}  \quad \Gamma_{ab}\in \mathbb{Z}.
\end{align}
Also, the constraint for additive codes of type I/II reads,
\begin{align}\label{eq:gamma_add_type}
	\sum_{j\geq 0} \sum_{\substack{E_{a} \in \E_n \\ \wt(E_a)=2j}} \Gamma_{aa}
	=\begin{cases}
2^{n-k-1} \quad \text{(Type I)} \\  2^{n-k} \quad \quad \!\text{(Type II)}\,.
	\end{cases}
\end{align}

\end{enumerate}

\begin{remark}
Note that of the original LP constraints~\eqref{eq:lp},
Eq.~\eqref{eq:sdpgamma} 
only contains non-negativity of $A_j \geq 0$ and the Knill-Laflamme conditions,
but not
$KB_i-A_i\geq 0$ and $S_i\geq 0$ for all $i$.
While they can, in principle, be added via the diagonal elements of $\Gamma$, they did not make any difference for the codes in our table.
\end{remark}

\begin{remark}\label{rmk:tr} Ref.~\cite{munne2025sdpboundsquantumcodes} lists 
$\tr(\Gamma)=2^n/K$ as additional SDP constraint. 
However, one can check that this constraint is redundant, as it can be obtained by combining
\begin{align}
\Gamma_{0b}=\Gamma_{bb}\,, \quad \sum^N_{b=0}\Gamma_{0b}=\frac{2^n}{K}\Gamma_{00}\,, \quad \Gamma_{00}=1\,.
\end{align}
\end{remark}

\subsection{Lovász SDP for self-dual codes}\label{sec:lov}
A relaxation of the main SDP
can be done in terms of the Lovasz number
for self-dual codes ($K=1$). 
This drops or modifies the last three conditions of Eq.~\eqref{eq:sdpgamma}:
a) the group structure imposed on $\Gamma$ is relaxed to a unitary and anti-commutativity structure, resulting in $\Gamma_{aa} = \Gamma_{a0}$ and $\Gamma_{ab} = 0$ if $E_a E_b = -E_b E_a$;
the constraint of the 
b) Knill-Laflamme conditions~\eqref{eq:klred} 
for the case of pure codes [Eq.~\eqref{eq:selfdual}]
are
applied;
c) the off-diagonal projector constraints are dropped.

In what follows, we consider simple graphs $G = (V,E)$ only, i.e. unweighted, undirected graphs without loops.
The size of the largest subset of vertices in $G$ in which no pair of vertices are adjacent is the independence number $\alpha(G)$.
In the context of coding theory, a classical code can be interpreted as a graph $G$ whose independence number $\alpha(G)$ equals the size of the code~\cite{1055985}.
Although computing this quantity is NP-complete~\cite{DBLP:books/fm/GareyJ79}, 
it is known that $\alpha(G) \leq \vartheta(G)$, where $\vartheta(G)$ is the Lovász number of a graph~\cite{1055985}.
Thus, the Lovász number can be used to upper bound the size of a classical code. 
It is known that $\vartheta$ can be computed in polynomial time via an SDP, 
and the symmetry-reduction of a slight variation of it 
leads to the classical Delsarte bound based on linear programming~\cite{1056072}.

A quantum analogue was derived in Ref.~\cite{munne2025sdpboundsquantumcodes} for self-dual quantum codes via the formulation in Eq.~\eqref{eq:sdpgamma}.
To see this, consider a self-dual quantum code
and 
define a graph $G$ whose vertices $V$ are those Pauli strings whose weights are larger or equal to the code distance, i.e. all $E_a\in V$ for which $\wt(E_a)\geq \delta$.
Let two vertices $a$ and $b$ be connected, written $a \sim b$, if $E_a,E_b \in V$ satisfies
$E_a E^\dag_b = -E_b E^\dag_a$ or $0<\wt(E_a E^\dag_b)<\delta$.
The Lovász theta number can be defined in terms of an SDP as~\cite{1055985,GALLI2017159},
\begin{align}\label{eq:lovasz_SDP2}
    \vartheta(G)=\max_M \quad & \tr(M) \nn \\
    \text{subject to} \quad & M_{aa}=m_a \quad\,\, \text{for all}\quad a \in V \,,  \nn\\
    & M_{ab}=0  \quad \,\,\quad\text{if}\quad a\sim b ,\nn \\
    & \Delta = \begin{pmatrix} 1 & m^T\\ m & M \end{pmatrix} \succeq 0 \,.
\end{align}

Ref.~\cite{munne2025sdpboundsquantumcodes} showed that if a self-dual quantum code with parameters $(\!(n,1,\delta)\!)_2$ exists, then the Lovász number of the graph $G$ satisfies $2^n \leq 1+\vartheta(G)$.
This is seen by noting that the structure of Eq.~\eqref{eq:lovasz_SDP2} is obtained by relaxing some conditions of Eq.~\eqref{eq:sdpgamma} in the case of self-dual codes. 
In particular, the matrix $\Delta$ is equivalent to the relaxed version of $\Gamma$, where all rows and columns corresponding to the elements with weight $0 < \wt(E_a) < \delta$ are removed.
That these rows and columns can be removed follows from $\Gamma \succeq 0$ together with Eq.~\eqref{eq:selfdual}. 
By Remark~\ref{rmk:tr}, self-dual codes satisfy $\tr(\Gamma)=2^n$ and thus, $\tr(\Delta)=\tr(M)+1$ must be at least $2^n$. 
Therefore, 
while the Lovász number bounds the {\em size} of classical codes, 
it also bounds the {\em existence} of self-dual quantum codes.

\subsection{Symmetry-reduced SDP}
For numerical purposes, the SDP~\eqref{eq:sdpx_relax} scales in manner 
far beyond current solver capabilities
(for details see Remark~\ref{rem:scaling}).
We thus use a symmetry-reduced version based on the non-binary Terwilliger algebra~\cite{GIJSWIJT20061719}.
We explain how this reduction works.
Recall that the quantum weight enumerators are obtained by averaging over all Pauli strings of a given weight [see Eq.~\eqref{eq:enum}].
In particular, 
the non-negative function $\tr(E\Pi)\tr(E^\dagger \Pi)$ is averaged over all Pauli strings of constant weight to obtain the 
$A$-enumerator.
In a similar spirit we average the matrix $\Gamma$ 
while keeping its positive semidefiniteness.

Denote by $s(E)$ the support of a Pauli string, that is, the subsystems it acts non-trivially on.
Denote by $\wt(E) = |s(E)|$ its weight. 
Now average the matrix $\Gamma$ as follows,
\begin{align}\label{eq:lambda_anycode}
    \lambda^{t,p}_{i,j}(\varrho)= \!\!\!\! \sum\limits_{
        \substack{E_a,E_b\in \E_n \\|s(E^\dagger_a)| \,=\, i
            \\ |s(E_b)| \,= \,j 
            \\ \lvert s(E^\dagger_a)\cap s(E_b)\rvert \,= \,t
            \\ \lvert s(E^\dagger_aE_b)\rvert \, = \, i+j-t-p \\}} \!\!\!\!\!\! \Gamma_{ab}(\varrho)  \,.
\end{align}
We call $\lambda^{t,p}_{i,j}(\varrho)$ the \emph{matrix weights}, and it can be checked that
\begin{equation}\label{eq:Alambda}
A_i(\varrho)=\lambda^{0,0}_{i,0}(\varrho)=\lambda^{i,i}_{i,i}(\varrho)
\end{equation} for all $i$.
Furthermore, $\lambda^{t,p}_{i,j}(\varrho)=\lambda^{t,p}_{j,i}(\varrho)$ since $\Gamma_{ab}(\varrho)=\Gamma_{ba}(\varrho)$.
Note that the range of valid $(i,j,t,p)$ for any $E_a, E_b\in \E_n$ is restricted to the set
\begin{equation}\label{eq:I-range}
 \II(n)= \{ (i, j, t, p)\,:\, 0 \leq p \leq t \leq i,j \quad \text{and} \quad i+j \leq t+n\, \}\,.
\end{equation}
Following the symmetry-reduction based on the Terwilliger algebra~\cite{GIJSWIJT20061719}, one has
\begin{equation} \label{eq:psd}
	\Gamma(\varrho) \succeq 0 \quad  \Longrightarrow \quad \bigoplus_{\substack{a,k \in \N_0\\ 0\leq a\leq k\leq n+a-k}}
	\left(\sum\limits_{\substack{t,p\in \N_0 \\ 0 \leq p \leq t \leq i,j \\ i+j\leq t+n}} \alpha(i,j,t,p,a,k)x_{i,j}^{t,p}(\varrho)\right)_{i,j=k}^{n+a-k} \succeq 0\,,
\end{equation}
where $x^{t,p}_{i,j}(\varrho)=\lambda^{t,p}_{i,j}(\varrho)/\gamma^{t,p}_{i,j}$ with  $\gamma^{t,p}_{i,j}$ and $\alpha(i,j,t,p,a,k)$ defined in Eq.~\eqref{eq:gamma} and \eqref{eq:alpha} respectively.
From constraints on 
$\Gamma(\varrho)$
one obtains constraints on $x^{t,p}_{i,j}(\varrho)$
via
Eq.~\eqref{eq:lambda_anycode}.
Then the SDP of Eq.~\eqref{eq:sdpgamma} 
can be relaxed to~\cite{munne2025sdpboundsquantumcodes},
\begin{align}\label{eq:sdpx_relax}
	\textnormal{find} 		\quad & x^{t,p}_{i,j}  \nn \\
	\textnormal{subject to} 	\quad
	&  x^{0,0}_{0,0}=1, \nn \\
    &x^{t,p}_{i,j}=0 \quad\quad\quad \text{if}\quad  t-p \quad  \text{is odd}\,,\nn \\
	& x^{t,p}_{i,j}=x^{t',p'}_{i',j'}\,\,\,\, \quad \text{if} \quad  t-p=t'-p' \quad \text{is even} \quad \text{and} \nn \\
	&\quad\quad\quad\quad\quad\quad\,\, (i,j,i+j-t-p) \quad\text{a permutation of} \quad(i',j',i'+j'-t'-p')\,, \nn \\
	& \sum_{\substack{(i,j,t,p)\in \II(n) \\ k=i+j-t-p}}  \gamma^{t,p}_{i,j} x^{t,p}_{i,j}=\frac{2^n}{K} \gamma^{0,0}_{k,0} x^{0,0}_{k,0}   \,, \nn\\
	& K2^{-n}\sum^n_{i=0} K_j(i,n) \gamma^{0,0}_{i,0} x^{0,0}_{i,0}\quad = \quad \gamma^{0,0}_{j,0} x^{0,0}_{j,0} \quad\quad \text{for} \quad 0<j<\delta \,, \nn \\
	&\bigoplus_{\substack{a,k \in \N_0\\ 0\leq a\leq k\leq n+a-k}} \left(\sum\limits_{\substack{t,p\in \N_0 \\ 0 \leq p \leq t \leq i,j \\ i+j\leq t+n}}\alpha(i,j,t,p,a,k)x_{i,j}^{t,p}\right)_{i,j=k}^{n+a-k}   \succeq 0\,,
\end{align}

\begin{remark}  \label{rem:scaling}  
The symmetry-reduction of the non-binary Terwilliger algebra for a $n^q\times n^q$ matrices 
(containing originally 
$O(n^{2q})$ real variables)
results in 
$O(n^4)$ real variables~\cite{GIJSWIJT20061719}.
In our case $q=4$, leading to a substantive reduction in the size of the resulting SDP.
\end{remark}

In contrast to the original formulation in Ref.~\cite[Theorem 22]{munne2025sdpboundsquantumcodes}, Eq.~\eqref{eq:sdpx_relax} does not include the constraint $\sum^n_{i=0}\gamma^{0,0}_{i,0}x^{0,0}_{i,0}=2^n/K$.
This is because this constraint is also redundant in the symmetry-reduced SDP (see Remark~\ref{rmk:tr}).

Additional constraints can be added to restrict the SDP~\eqref{eq:sdpx_relax} further. 
Writing $\lambda^{t,p}_{i,j}= x^{t,p}_{i,j}\gamma^{t,p}_{i,j}$, these constraints reads

\begin{enumerate} [label=(\alph*)]\setlength{\itemsep}{3pt}
\item {\bf Pure codes.} 
For all $x^{t,p}_{i,j}$,
\begin{align}\label{eq:klred_pure_xijtp}
x^{t,p}_{i,j}=0 \quad\quad \text{if} \quad\quad \{i,j,i+j-t-p\} \cap \{1,\cdots, \delta-1\} \neq \emptyset\,.
\end{align}
\item {\bf Self-dual codes.} Self-dual codes are pure, and thus the condition in (a) applies. Additionally
\begin{align}\label{eq:seldualred}
 \sum^n_{i=0} (-1)^i K_j(i,n) \lambda^{i,i}_{i,i} &=0 \quad\quad \quad  \text{if} \quad  n-j \quad \text{is odd}\,, \nn\\
\frac{K}{2^n}\sum^n_{i=0} K_j(i,n) \lambda^{i,i}_{i,i}
 &= 
	\lambda^{j,j}_{j,j}  \quad \quad \text{for} \quad 0\leq j \leq n \,,
\end{align}
\item {\bf Additive codes}. For all $x^{t,p}_{i,j}$,
\begin{align}\label{eq:add_sym_red}
x^{t,p}_{i,j} \geq 0 \quad\quad \text{and}  \quad\quad \lambda^{t,p}_{i,j}\in \mathbb{Z}\,.
\end{align}
For type I/II additive codes, one can further impose
\begin{align}\label{eq:add_type_symred}
	\sum_{j\,\text{even}} \lambda^{j,j}_{j,j} =
	\begin{cases}
2^{n-k-1} \quad \text{(Type I)} \\  2^{n-k} \quad \quad \!\text{(Type II)}\,.
	\end{cases}
\end{align}

\end{enumerate}

\begin{remark}\label{rmk:extraconst}
Additional constraints on 
$x^{0,0}_{i,0}\gamma^{0,0}_{i,0}$
could be added to the SDP~\eqref{eq:sdpx_relax} that arise from the linear programming bound, namely
$KB_i-A_i\geq 0$, 
and 
$S_i\geq 0$ for all $i$. 
Also, further SDP constraints 
(see Ref.~\cite[Theorem 22]{munne2025sdpboundsquantumcodes})
could be added to Eq.~\eqref{eq:sdpx_relax}.
For the numerical results in Section~\ref{sec:cert}, 
these additional LP and SDP constraints did not make a difference in the bounds obtained.
\end{remark}

\begin{example}
Consider the matrix $\Gamma(\varrho)$ of the five-qubit code $[\![5,1,3]\!]_2$ as defined in Example~\ref{ex:enum} and with $\varrho=\Pi/2$.
Its matrix weight enumerator $\lambda^{t,p}_{i,j}$ with $(i,j,t,p)$ in
\begin{equation}
 \II(5)= \{ (i, j, t, p)\,:\, 0 \leq p \leq t \leq i,j \quad \text{and} \quad i+j \leq t+5\, \}\,,
\end{equation}
is then given by
\begin{align}
\lambda^{0,0}_{0,0}(\varrho)= 1\,,\quad 
\lambda^{0,0}_{4,0}(\varrho) = \lambda^{4,4}_{4,4} (\varrho)= 15\,, \quad \lambda^{3,1}_{4,4}(\varrho)= 180\,, \quad
\lambda^{4,0}_{4,4} (\varrho)= 30\,,
\end{align}
with all other terms vanishing,
$\lambda^{0,0}_{1,0} = \lambda^{0,0}_{1,1} =\dots = 0$.
\end{example}

We have the following open problem regarding a putative additive code for which a possible integral weight enumerator satisfying all LP constraints is known, but for which we do not know whether also integral {\it matrix weights} exist.
\begin{problem}
Consider the hypothetical $[\![24,0,10]\!]_2$ stabilizer code. A putative weight enumerator is~\footnote{
This weight distribution can also be found in the Online Encyclopedia of Integer Sequences \href{A030331}{http://oeis.org/A030331}}
\begin{align}
&(A_{10}, A_{12}, A_{14}, \cdots,A_{24})  \nn\\ 
&\quad= 
(18216, 156492, 1147608, 3736557, 6248088, 4399164, 1038312, 32778)\,.
\end{align}
The problem asks to provide a putative integer matrix weight enumerator $\lambda^{t,p}_{i,j}$ for this code.
\end{problem}

\section{Code bounds}
\label{sec:cert}
\subsection{Methods}
To certify the infeasibility of a primal problem, 
the dual of an SDP can be used:
Consider primal minimization and dual maximization problems $p$ and $d$.
By weak duality, $p \geq d$ for any pair of primal and dual feasible objective values~\cite{boyd2004convex}. 
For testing primal feasibility problem one sets $p=0$,
and as a consequence, a strictly positive dual solution 
$d > 0$ shows primal infeasibility.
In our case, a solution of Eq.~\eqref{eq:sdpx_relax} with positive objective value implies that a quantum code 
with parameters $(\!(n,K,\delta)\!)_2$ does not exist.

To obtain reliable infeasibility certificates, we employ the low-rank SDP solver by Cohn, Leijenhorst, and de Laat~\cite{cohn2024optimalitysphericalcodesexact, Leijenhorst_2024}, online available at \url{https://github.com/nanleij/ClusteredLowRankSolver.jl}.
Originally developed to compute bounds on spherical codes, this high-precision primal-dual interior-point solver can exploit both low-rank structures and the clustering of constraints in semidefinite programs.
The solver comes with a heuristic rounding method for extracting solutions over rational or algebraic numbers.

To verify positive semidefiniteness of the obtained solutions, 
we rely on matrix diagonalization\footnote{
Another method relies on the characteristic polynomial: 
a symmetric $n\times n$ 
matrix $A$ is positive semidefinite, 
if and only if 
the coefficients of its characteristic polynomial
$p(\lambda) = \det(\lambda I -A) = \lambda^n + p_{n-1}\lambda^{n-1} + 
\dots + p_0$ alternate in sign. That is,
$(-1)^{n-i}p_i\geq 0$ for all $1\leq i\leq n$.}~\cite{PEYRL2008269}. 
The matrix $A = LDL^T$ is positive semidefinite, 
if and only if the diagonal matrix $D$ contains only non-negative entries.
Note that the $LDL^T$ decomposition relies 
on basic arithmetic operations only.
It can thus be computed exactly 
in algebraic extensions of rational fields.
The symmetry-reduced SDP in Eq.~\eqref{eq:sdpx_relax} involves the constants $\alpha(i,j,t,p,a,k)$, 
whose expansion contains powers of $\sqrt{q-1}$. 
In our case $q=4$ and we perform an $LDL^T$ decomposition 
over $\Q(\sqrt{3})$.
The elements of $D$ are thus of form $a+bz \in \Q(\sqrt{3})$ where $a,b\in \Q$ and $z=\sqrt{3}$.
We check the non-negativity of these elements can be done by approximating $\sqrt{3}$ by two rational numbers such that $r_\ell < \sqrt{3} < r_u$ with $r_u,r_\ell \in \mathbb{Q}$. 
Then,
\begin{align}
    a+br_\ell \geq 0 \quad \text{implies} \quad a+b\sqrt{3} \geq 0 \quad \text{for}\quad 
    \tfrac{b}{a}
    \geq 0\,, \nn\\
    a+br_u \geq 0 \quad \text{implies} \quad a+b\sqrt{3} \geq 0 \quad \text{for}\quad 
    \tfrac{b}{a}\leq 0\,, \nn\\
    b \geq 0 \quad\, \text{implies} \quad a+b\sqrt{3} \geq 0 \quad \text{for}\quad 
    a=0 \,.
\end{align}

In the following, we solve the dual of four variants of the symmetry-reduced SDP. All dual programs can be found in Appendix~\ref{app:dual}.
\begin{enumerate}[label=(\alph*)]
\setlength{\itemsep}{3pt}

\item For {\bf general codes} with $K>1$, we use SDP~\eqref{eq:sdpx_relax} whose dual is SDP~\eqref{eq:dual_solall}.

\item 
For {\bf pure codes} with $K>1$, we use SDP~\eqref{eq:sdpx_relax}
with the additional constraints of
Eq.~\eqref{eq:klred_pure_xijtp}.
The corresponding dual SDP is Eq.~\eqref{eq:dual_solall} 
with the modifications described in 
Proposition~\ref{rmk:dualpure}.

\item For {\bf self-dual codes} (pure and $K=1$), we relax the SDP for pure codes as shown in Appendix~\ref{app:selfdual}. Its dual is Eq.~\eqref{eq:dual_sol}. For the range of codes studied in this work, $n\leq 19$, this relaxation gives the same bounds as the pure code SDP, 
but it is easier to compute.

\item  For {\bf additive codes} with $K>1$, we use SDP~\eqref{eq:sdpx_relax} with addition  additive constraints in Eq.~\eqref{eq:add_sym_red} and
~\eqref{eq:add_type_symred}, with the exception of $\lambda^{t,p}_{i,j} = x^{t,p}_{i,j}\gamma^{t,p}_{i,j}\in \mathbb{Z}$ . 
Its dual is SDP ~\eqref{eq:dual_solall} 
with the modifications described in 
Proposition~\ref{rmk:additive codes}.
In the parameter range considered,
($n \leq 19$),
these constraints did not lead to improvements on bounds for additive codes.

\end{enumerate}

The four dual SDPs could further be constrained by applying additional constraints that hold for all quantum codes, 
namely  
$KB_i \geq A_i$ and $S_i \geq 0$,
and the additional SDP variable from Ref.~\cite{munne2025sdpboundsquantumcodes}[Theorem 22].
However, including these extra conditions
did not lead to stronger bounds for the parameter range ($n \leq 19$) considered in this work.

All infeasibility certificates can be found at \url{https://github.com/ganglesmunne/SDP_bounds_on_quantum_codes_rational_certificates}.

\begin{table}[tbp]\label{tab:results}
\vspace{0.3cm}
{\footnotesize
\begin{tabular}{@{} c r r r r r r r @{}}
\toprule 
n $\backslash$ $\delta$ & 2 & 3 & 4 & 5 & 6 & 7 & 8 \\ 
\midrule
6 & 16 & $2^\alpha$ & 1 & 0 & 0 & 0 & 0 \\
7 & 24 - 26 & 2 - 3 & 0 & 0 & 0 & 0 & 0 \\
8 & 64 & \textbf{8}(9) & 1 & 0 & 0 & 0 & 0 \\
9 & 100 - 112 & 12 - 13& 1 & 0 & 0 & 0 & 0  \\
10 & 256 & 24 & \textbf{4}(5) & 0 & 0 & 0 & 0 \\
11 & 416 - 460 & 32 - $\bm{42}^\alpha$(53) & 4 - 7 & 2 & 0 & 0 & 0 \\
12 & 1024 & 64 - 89 & 16 - 20 & 2 & 1 & 0 & 0\\
13 & 1586 - 1877 & 128 - 204 & 20 - 40 & 2 - 3 & \textbf{0}(1) & 0 & 0\\
14 & 4096 & 256 - $\bm{295}^\alpha$(324) & 64 - 102 & 4 - $\bm{9}^\alpha$(10) & 1 & 0 & 0\\ 
15 & 6476 - 7606 & 512 - 580 & 64 - $\bm{138}^\alpha$(150) & 8 - 18 & 1 & 0 & 0\\
16 & 16384 & 1024 - \textbf{1170}(1260) & 128 - 256 & 32 - 42 & 4 - 6 & 0 & 0 \\
17 & 26333 - 30720 & 2048 - 2145 & 512 & 32 - $\bm{59}^\alpha$(71) & 4 - 8 & 2 & 0\\
18 & 65536 & 2048 - 4096 & 512 - \textbf{921}(986) & 64 - 113 & 16 - 22 & 2 & 1 \\
19 & 106762 - 123790 & 4096 - $\bm{7420}^\alpha$(8426) & 1024 - 1804 & 128 - \textbf{249}(276) & 32 - 47 & 2 - 3 & \textbf{0}(1) 
\\
\bottomrule
\end{tabular}
}
\vspace{1cm}
 \caption{{\bf Maximum size $K$ of quantum codes}
given $n$ and $\delta$.
Lower and upper bounds are separated by a dash $-$. 
Upper and lower bounds meet
if only one number is shown.
A $K=0$ implies that no corresponding code exists.
When our SDP improves a known bound, the stronger SDP bound is highlighted in bold with the previous bound in parentheses. Codes labeled by $\alpha$ must be impure; stronger bounds for pure codes are provided in Eq.~\eqref{eq:pure_inf}.}
\end{table} 

\subsection{Table}

Table~\ref{tab:results} shows 
upper and lower bounds on the 
maximum size $K$ of non-additive quantum codes given $n$ and $\delta$
All upper bounds are obtained via a rational certificate.
Lower and upper bounds are separated by a dash $-$. 
Upper and lower bounds meet
if only one number is shown.
A $K=0$ implies that no corresponding code exists.
If our SDP improves a known upper bound, the previous bound is shown in parentheses, with the SDP bound highlighted in bold.
Codes labeled by $\alpha$ must be impure; stronger bounds for pure codes are provided in Eq.~\eqref{eq:pure_inf}.

Our SDP bounds are compared to previously known bounds.
For 
$6\leq n \leq 15$ and 
$2\leq \delta \leq  5$, the known upper and lower bounds are obtained from combining the following results: 
\begin{itemize}\setlength{\itemsep}{3pt}
    \item[(i)] Ref.~\cite{munne2025sdpboundsquantumcodes} showed the non-existence of $(\!(8,9,3)\!)_2$ and $(\!(10,5,4)\!)_2$ numerically via SDP, 
    for which we now also provide exact certificates.
    \item[(ii)]     Ref.~\cite{Rigby_2019} found the following non-additive codes $(\!(9,100,2)\!)_2$, $(\!(11,416,2)\!)_2$, and $(\!(13,20,4)\!)_2$ serving as lower bounds.
    \item[(iii)]     The remaining bounds are from
    ~\cite[Table III]{Rigby_2019}.
\end{itemize} 

\smallskip
Known Bounds for $n\geq 16$ are obtained from:
\begin{itemize}\setlength{\itemsep}{3pt}
    \item[(i)] For $\delta=2$
    and $n=16,17,18,19$,
    lower bounds 
 are given by the family of non-additive
$(\!(4k +2l +3, M_{k,l}, 2)\!)_2$
 codes 
 of Ref.~\cite{PhysRevLett.99.130505},
where 
$M_{k,l} \approx 2^{n-2}(1-\sqrt{2/(\pi(n - 1))})$.
 
 \item[(ii)] For $\delta>2$, lower bounds are from Ref.~\cite{table}.
 
 \item[(iii)] Upper bounds are computed using the LP of Eq.~\eqref{eq:lp}.
\end{itemize}
Bounds with $n\leq6$ are not shown since they are tight.
In contrast to the online tables by Ref.~\cite{table}, we also list non-additive codes and arrange columns and rows by $n\backslash\delta$ instead of $n\backslash k$ with $k=\log_2(K)$. 

\bigskip
We list comments on specific cases below:

\smallskip
{\bf Improvements upon certificates:}
\begin{enumerate}
\setlength{\itemsep}{3pt}
\item[(a)]
High-precision infeasibility certificates for 
$(\!(7,1,4)\!)_2$,
$(\!(8,9,3)\!)_2$,
and 
$(\!(10,5,4)\!)_2$ codes
were provided in Ref.~\cite[Appendix B]{munne2025sdpboundsquantumcodes}. 
We now provide exact certificates.
Note that the non-existence of a $(\!(7,1,4)\!)_2$ code was already known analytically~\cite{PhysRevLett.118.200502}.

\end{enumerate}

\smallskip
{\bf Additive/non-additive codes:}
\begin{enumerate}\setlength{\itemsep}{3pt}

\item[(b)] 
Additive $[\![13,0,6]\!]_2$ and $[\![19,0,8]\!]_2$ codes do not exist~\cite{681315}. 
We also show that the corresponding 
{\it non-additive} codes
$(\!(13,1,6)\!)_2$ and $(\!(19,1,8)\!)_2$
do not exist.

\item[(c)] 
An additive $[\![17, 6, 5]\!]_2$ does not exist due to the improved upper bound $(\!(17,59,5)\!)_2$.
It is known from the LP bounds that a 
$[\![19, 2^8, 5]\!]$ stabilizer code does not exist. 
We can now show that the same result (already anticipated
by Calderbank et al.~\cite{681315}~\footnote{
"However, we will be surprised if a $(\!(19, 2^8, 5)\!)$ non-additive code exists."~\cite{681315}})
holds for non-additive codes: 
a $(\!(19, 2^8, 5)\!)_2$ does not exist due to 
the improved upper bound $(\!(19, 249, 5)\!)_2$. 

\end{enumerate}

\smallskip
{\bf Pure codes:}
\begin{enumerate}\setlength{\itemsep}{3pt}
\item[(d)] For pure codes the bounds strengthen to
\begin{align}\label{eq:pure_inf}
 &(\!(6,1,3)\!)_2 \,,  \quad  (\!(11,41,3)\!)_2 \,,  \quad  (\!(14,290,3)\!)_2 \,,  \nn \\
&(\!(14,8,5)\!)_2 \,,   \quad  (\! (15,135,4)\!)_2\,, \quad (\!(17,57,5)\!)_2 \,, \quad (\!(19,7314,3)\!)_2 \,.
\end{align}
\end{enumerate}

\begin{enumerate}
\item[(e)]
Additive $[\![6,1,3]\!]_2$ codes must be impure~\cite{681315}.
Our SDP extends this result, showing that the corresponding non-additive $(\!(6,2,3)\!)_2$ codes must also be impure.

\item[(f)] Ref.~\cite{681315} proved analytically that pure additive $[\![12,1,5]\!]_2$ codes do not exist.
While our SDP cannot show this,
one can hope that a mixed-integer SDP solver could recover this result:
for additive codes the SDP is further restricted by $\lambda^{t,p}_{i,j} = x^{t,p}_{i,j}\gamma^{t,p}_{i,j}\in \mathbb{Z}$ [see Eq.~\eqref{eq:add_sym_red}].
\end{enumerate}

\bibliographystyle{alpha}
\bibliography{current_bib}

\appendix
\section{Terwiliger Algebra factors}\label{app:Terw}

In the block-diagonalization of the Terwilliger algebra, one has the following normalization factors~\cite{GIJSWIJT20061719}:
\begin{align}\label{eq:gamma}
	\gamma_{i,j}^{t,p}
	&= (q-1)^{i+j-t}(q-2)^{t-p}\binom{n}{a_1,\dots,a_r}\frac{n!}{a_1!\dots a_r!(n-\sum^r_{\ell=1} a_\ell)!}\,, 
    \\\nn\\\nn
	\alpha(i,j,t,p,a,k) &= \beta_{i-a,j-a,k-a}^{n-a,t-a}\left(q-1\right)^{\frac{1}{2}\left(i+j\right)-t} \sum\limits_{g=0}^{p}\left(-1\right)^{a-g}\binom{a}{g}\binom{t-a}{p-g}\left(q-2\right)^{t-a-p+g}\,,\label{eq:alpha}\\
	\beta_{i,j,k}^{m,t}&=\sum\limits_{u=0}^{m}\left(-1\right)^{t-u}\binom{u}{t}\binom{m-2k}{m-k-u}\binom{m-k-u}{i-u}\binom{m-k-u}{j-u}\,.
\end{align}

\section{Dual programs}\label{app:dual}

Here we show the three dual SDPs used to obtain the infeasibility certificates for general, pure, and self-dual quantum codes.
These certificates are in~\url{https://github.com/ganglesmunne/SDP_bounds_on_quantum_codes_rational_certificates} and prove the new upper bounds in the maximum size of a quantum code shown in Section~\ref{sec:cert}.

\subsection{General codes}
The dual of SDP~\eqref{eq:sdpx_relax} is~\cite{munne2025sdpboundsquantumcodes},
\begin{align}\label{eq:dual_solall}
	\alpha \quad = \quad \max_{Y^{(a,k)},\, C_i} \quad &   (KD_0-C_0)  - {y^{0,0}_{0,0}} \quad &  \nn \\
	\text{subject to}  \quad & Y^{(a,k)} \succeq0,\nn \\
	& y^{t,p}_{i,j}\ = \frac{1}{\gamma^{t,p}_{i,j}} \sum^{\min(i,j)}_{k=0} \sum^k_{a=\max(i,j)+k-n} \alpha(i,j,t,p,a,k) Y^{(a,k)}_{i-k,j-k}\,, \nn \\
	&D_i = 2^{-n} \sum^{\delta-1}_{j=0} K_j(i,n)  C_{j} \,, \nn \\
	&Q_i=\frac{1}{\left(\frac{2^n}{K}-2\right)}   (2y^{0,0}_{i,0} + y^{i,i}_{i,i} -KD_i + C_i)
	\quad\quad\quad\quad\quad\hspace{-0.03cm} \text{for} \quad 0 <i <\delta \,,\nn \\
	& Q_i=\frac{1}{\left(\frac{2^n}{K}-2\right)}   (2y^{0,0}_{i,0} + y^{i,i}_{i,i} -KD_i)
	\quad\quad\quad\quad\quad\quad\quad\, \text{for} \quad \delta \leq i \leq n \,, \nn \\
	& \sum_{\substack{ (i',j',t',p')\in \II(4,n) \\ (t-p)=(t'-p') \\ (i',j',i'+j'-t'-p') \, \text{a} \\
			\, \text{permutation of} \,  (i,j,i+j-t-p) \\ }} \!\!\!\!\!\!\!\!\!\!  
	\left(y^{t',p'}_{i',j'}+ Q_{i'+j'-t'-p'}\right) = 0 \quad  \hspace{2.3em} \text{if} \quad (t-p) \quad \text{is even}  \nn \\ &  \vspace{-10cm}\hspace{6cm}
    \quad\quad\quad\quad \text{and} \quad i,j,i+j-t-p \neq 0 \,.
\end{align}
Here $Y^{(a,k)}$ are real square matrices of size $(n+a-2k)$ with $0\leq a\leq k\leq n+a-k$.
Due to Remark~\ref{rmk:tr}, Eq.~\eqref{eq:sdpx_relax} contains one constraint less,
namely $\sum^{n}_{i=0}\gamma^{0,0}_{i,0} x^{0,0}_{i,0}=2^n/K$,
than 
the original SDP formulation in Ref.~\cite{munne2025sdpboundsquantumcodes}
since it is redundant. 
From SDP duality, every primal constraint corresponds to a dual variable and vice versa.
The redundant primal constraint that is removed in Eq.~\eqref{eq:sdpx_relax} leads to the removal of the dual variable $w$ in Ref.~\cite[Proposition 26]{munne2025sdpboundsquantumcodes}.
This leads to the SDP in Eq.~\eqref{eq:dual_solall}.

\subsection{Pure codes}
For pure codes the SDP~\eqref{eq:dual_solall} is modified.
\begin{proposition}\label{rmk:dualpure}
    For pure codes 
    the dual of SDP 
~\eqref{eq:sdpx_relax} 
with the additional constraint
Eq.~\eqref{eq:klred_pure_xijtp} 
consists of SDP~\eqref{eq:dual_solall}
with the following modifications:

\begin{itemize}
    \item[(i)] The variable $Q_i$ is indexed by $\delta \leq i \leq n$ 
    instead of $0 \leq i \leq n$.
    \item[(ii)] Last constraint in SDP~\eqref{eq:dual_solall} 
    is only applied when $(i,j,t,p)$ satisfies that $(t-p)$ is even and $i\geq \delta$ and $j\geq \delta$ and $i+j-t-p\geq \delta$.
\end{itemize}

\end{proposition}
\begin{proof}
For pure codes the conditions of Eq.~\eqref{eq:klred_pure_xijtp} can be equivalently written as: $x^{t,p}_{i,j} = 0$ if $0<i<\delta$ or $0<j<\delta$ or $0<i+j-t-p<\delta$. 
By SDP duality, each primal variable corresponds to a dual constraint.
Thus, all constraints from the dual corresponding to the removed primal variables can be omitted. 
In our case, (i) the last equality of Eq.~\eqref{eq:dual_solall} is only applied if $(t-p)$ is even, $i\geq \delta$, $j\geq \delta$, and $i+j-t-p\geq \delta$.
Furthermore, (ii) $Q_i$ is now only defined when $\delta \leq i \leq n$.  This ends the proof.
\end{proof}

\begin{proposition}
For pure codes, the SDP variables of Eq.~\eqref{eq:sdpx_relax} can be simplified 
by removing all rows and columns containing $x^{t,p}_{i,j}$ with $0<i<\delta$ or $0<j<\delta$ .
As a consequence, 
the size of the dual variables $Y^{(a,k)}$ reduces from $(n+a-2k)$ to 
\begin{align}
   (n+a-2k-\delta+2) \quad  &\text{for}  \quad k=0\,\, (a=0),\nn\\
    (n+a-2k-\delta+1) \quad  &\text{for}  \quad k\leq \delta \quad \text{and} \quad (n-a-k)\geq \delta
\end{align}
for pure codes (Proposition~\ref{rmk:dualpure}).
For other values of $(a,k)$, the matrix $Y^{(a,k)}$ is not reduced.
\end{proposition}
\begin{proof}
For pure codes, the constraints in Eq.~\eqref{eq:klred_pure_xijtp} set $x^{t,p}_{i,j}=0$ for $0< i<\delta$ or $0<j<\delta$. As a consequence, all  rows and columns of the SDP variables containing $x^{t,p}_{i,j}$
with $0< i<\delta$ or $0<j<\delta$ are zero and thus, all these rows and columns can be removed. Since the size of the primal SDP constraints of Eq.~\eqref{eq:sdpx_relax} is reduced, the size of the corresponding dual variable $Y^{(a,k)}$ is accordingly reduced. This ends the proof.
\end{proof}

\subsection{Self-dual codes}\label{app:selfdual}

Here we relax the symmetry-reduced SDP for pure codes, i.e. SDP~\eqref{eq:sdpx_relax} with the additional constraints Eq.~\eqref{eq:klred_pure_xijtp}.
This relaxation is easier to compute and gives the same bounds for self-dual codes as the SDP pure code when $n\leq 19$.
The relaxation is obtained by the following 
modifications:
\begin{itemize}\setlength{\itemsep}{3pt}
    \item[(i)] The condition
    \begin{align}
        &x^{t,p}_{i,j}=x^{t',p'}_{i',j'}\,\,\,\, \quad \text{if} \quad  t-p=t'-p' \quad \text{is even} \quad \text{and} \nn \\
	&\quad\quad\quad\quad\quad\quad\,\, (i,j,i+j-t-p) \quad\text{a permutation of} \quad(i',j',i'+j'-t'-p')\,, \nn \\
    \end{align}
    is only applied in the case 
    $x^{i,i}_{i,i}=x^{0,0}_{i,0}$ for $0\leq i\leq n$.
    \item[(ii)] 
    The Knill-Laflamme conditions,
     \begin{align}
            K2^{-n}\sum^n_{i=0} K_j(i,n) \gamma^{0,0}_{i,0} x^{0,0}_{i,0}\quad = \quad \gamma^{0,0}_{j,0} x^{0,0}_{j,0} \quad\quad \text{for} \quad 0<j<\delta \,,
     \end{align}
    are dropped.
    \item[(iii)] The projector constraints,
    \begin{align}
    \sum_{\substack{(i,j,t,p)\in \II(n) \\ k=i+j-t-p}}  \gamma^{t,p}_{i,j} x^{t,p}_{i,j}=2^n \gamma^{0,0}_{k,0} x^{0,0}_{k,0}\,,
    \end{align}
    are only applied in the case $\sum^n_{i=0} \gamma^{0,0}_{i,0} x^{0,0}_{i,0}=2^n$, 
    where we used the fact that $\gamma^{0,0}_{0,0} x^{0,0}_{0,0}=1$.    
\end{itemize}

The dual of 
~\eqref{eq:sdpx_relax}
with modifications (i) -- (iii) is~\cite{munne2025sdpboundsquantumcodes},
\begin{align}\label{eq:dual_sol}
	\alpha \quad = \quad \max_{Y^{(a,k)}} \quad
	&   (2^n-1)w - y_{0,0}^{0,0} \nn \\
	\text{subject to}  \quad & Y^{(a,k)} \succeq 0\,, \nn \\
	& y^{t,p}_{i,j} = \frac{1}{\gamma^{t,p}_{i,j}}\sum^{\min(i,j)}_{k=0} \sum^k_{a=\max(i,j)+k-n} \alpha(i,j,t,p,a,k) Y^{(a,k)}_{i-k,j-k}\,, \nn \\
    & w = -y_{n,n}^{n,n}  - 2 y_{n,0}^{0,0}\,, \nn \\
	& w-y_{i,i}^{i,i} -2 y_{i,0}^{0,0} = 0
	\quad\quad\quad\text{for} \quad \delta \leq i < n \,, \nn \\
	&  y^{t,p}_{i,j}=0
	\hspace{8em}  \,\text{if} \quad i,j\neq 0 \quad \text{and}\quad t-p \quad  \text{is even } \quad  \nn \\
	&\hspace{11.6em}\quad \text{and} \quad i + j - t - p \geq \delta \,,
\end{align}
where $Y^{(a,k)}$ with $0\leq a\leq k\leq n+a-k$ are real square matrices of size $(n+a-2k)$.

This SDP relaxation can also be interpreted as a symmetry-reduced version of the Lov\'asz SDP~\eqref{eq:lovasz_SDP2}. 

\subsection{Additive codes}
While further constraints can be added to the dual SDP~\eqref{eq:dual_solall} in the case of additive codes, these did not have any impact for the parameter range considered here, $n\leq 19$.
For completeness, we list the modifications below.
\begin{proposition}\label{rmk:additive codes}
    Consider the SDP~\eqref{eq:dual_solall} 
    with the additional code constraints of Eq.~\eqref{eq:add_sym_red} and ~\eqref{eq:add_type_symred}
    except
    the integer constraint.
    Its dual is Eq.~\eqref{eq:dual_solall} with the following modifications:
\begin{itemize}
    \item[(i)] For even $i$, the variable $D_i$ is redefined as
    \begin{align}
    D_i = 2^{-n} \sum^{\delta-1}_{j=0} K_j(i,n)  C_{j} +\begin{cases}  2^{-n+1}  K_0(i,n)G    \quad \text{(Type I)} \,, \\  2^{-n}K_0(i,n)G \quad\hspace{0.32cm} \text{(Type II)} \,,\end{cases} 
    \end{align}
    \item[(ii)] The objective function now reads as $(KD_0 - C_0 - G) -y^{0,0}_{0,0}$,
    \item[(iii)] The variable $y^{t,p}_{i,j}$ is redefinied as
    \begin{align}
    y^{t,p}_{i,j}\ = \left(\frac{1}{\gamma^{t,p}_{i,j}} \sum^{\min(i,j)}_{k=0} \sum^k_{a=\max(i,j)+k-n} \alpha(i,j,t,p,a,k) Y^{(a,k)}_{i-k,j-k} \,,\right) + g^{t,p}_{i,j}\,, \nn \\
    \end{align}
 \end{itemize}
 where $G \in \R$ and $g^{t,p}_{i,j}\geq 0$ are new dual variables.
 \begin{proof} 
 Eq.~\eqref{eq:add_type} can be rewritten as 
\begin{align}
    \gamma^{0,0}_{0,0} x^{0,0}_{0,0}=\begin{cases}  K2^{-n+1}\sum_{i \,\text{even}} K_0(i,n)\gamma^{0,0}_{i,0} x^{0,0}_{i,0}  \quad \text{(Type I)} \,, \\   K2^{-n}\sum_{i \,\text{even}} K_0(i,n)\gamma^{0,0}_{i,0} x^{0,0}_{i,0} \quad\hspace{0.32cm} \text{(Type II)} \,.\end{cases} 
\end{align}
since $\gamma^{0,0}_{0,0} x^{0,0}_{0,0}=1$ and $K_0(i,n)=1$ for all $i$. This constraint will correspond to a dual variable $G$ and will modify the definition of $D_i$ (i) and the objective function (ii).
For (iii), the constraints $x^{t,p}_{ij}\geq 0$ leads to adding an extra nonnegative dual variables $g^{t,p}_{i,j}$ to the definition of $y^{t,p}_{i,j}$.  This ends the proof.
 \end{proof}
\end{proposition}

\end{document}